\documentclass[11pt]{article}

\usepackage[margin=1in]{geometry}
\usepackage{amsmath,amssymb,amsthm,bm}
\usepackage{booktabs}
\usepackage{graphicx}
\usepackage{hyperref}
\usepackage[numbers]{natbib}
\usepackage{float}
\usepackage{setspace}
\usepackage[multiple]{footmisc}
\usepackage[left]{lineno}

\newtheorem{theorem}{Theorem}

\title{ Model Validation of Agentic AI Systems: \\ A POMDP-Based Framework for Belief-State, Forecast, and Policy Validation } \author{ Matthew Dixon\\ Quiota LLC } \date{\today} 
\begin{document}
\maketitle 
\begin{abstract}
Agentic artificial intelligence systems introduce a new class of model risk. Unlike traditional predictive models, autonomous agents continuously acquire information, form beliefs regarding latent states of the environment, generate forecasts, select actions, and adapt their behavior over time. Existing validation methodologies focus primarily on predictive accuracy and therefore provide limited insight into the quality of the underlying decision process. This paper proposes a model validation framework for agentic AI based on Partially Observable Markov Decision Processes (POMDPs). The framework decomposes autonomous decision making into information, beliefs, forecasts, actions, and utility, allowing each component to be validated independently. Large language models (LLMs) are formalized as approximate Bayesian filtering operators, and a model-risk taxonomy is developed encompassing state-space, filtering, forecast, policy, utility-specification, and parameter risks.

The model risk validation methodology is demonstrated through a portfolio-management case study in which an agent infers latent market regimes from market and macroeconomic information, generates belief-conditioned forecasts, and constructs portfolios using a Black--Litterman framework. Empirical validation combines performance analysis, belief calibration diagnostics, coverage tests, ablation studies, and parameter-sensitivity analysis. The results indicate that latent-state inference contributes independently to decision quality and that the principal conclusions remain robust across a broad range of parameter values. The principal contribution of the paper is a practical framework for extending established model risk management concepts to autonomous AI systems and providing a rigorous foundation for their validation, governance, and monitoring.

\end{abstract}

\section{Introduction}
Modern agentic artificial intelligence systems differ fundamentally from traditional predictive models. Rather than generating a single forecast or classification, autonomous agents continuously acquire information, form beliefs regarding latent states of the environment, generate forecasts, select actions, and adapt their behavior in response to feedback. Recent advances in large language models, foundation models, and autonomous agents have accelerated the deployment of such systems across finance, cybersecurity, enterprise operations, and scientific research \citep{bommasani2021,bai2022constitutional,weidinger2022,liang2022holistic,park2023,yao2023react,shinn2023reflexion,xi2023,wang2024agents}. As decision making becomes increasingly delegated to AI systems, the challenge is no longer merely one of predictive accuracy. Instead, the central question becomes whether an autonomous system can be trusted to make decisions under uncertainty.

This challenge naturally raises a model validation problem. Traditional machine-learning validation focuses primarily on predictive performance through metrics such as classification accuracy, mean squared error, perplexity, or benchmark task-completion rates. Such measures are appropriate when the objective is prediction. They become considerably less informative when the objective is sequential decision making. An autonomous system may produce accurate forecasts while making poor decisions, or alternatively achieve favorable outcomes despite possessing poorly calibrated beliefs. Consequently, validation must extend beyond outputs and address the quality of the underlying decision process itself.

The broader model-risk literature has long recognized that model performance cannot be evaluated solely through predictive accuracy. Regulatory guidance such as SR 11-7 and BCBS 239 emphasizes conceptual soundness, ongoing monitoring, governance, challenge processes, and the management of model uncertainty as essential components of model validation \citep{sr117,bcbs239}. Similarly, recent work in stress testing, scenario analysis, probability-of-default validation, and model-risk quantification has highlighted the importance of robustness, calibration, and fitness for purpose \citep{skoglund2019,zhang2019,rubtsov2016}. These concerns become even more important in autonomous systems because errors may arise not only from inaccurate forecasts but also from incorrect beliefs, flawed decision policies, inappropriate objectives, or deficiencies in information processing.

Many of these challenges have already been studied in stochastic control, Bayesian decision theory, and quantitative finance. Dynamic programming, Markov decision processes, and reinforcement learning provide mathematical frameworks for sequential decision making under uncertainty \citep{bellman1957,bertsekas1995,puterman1994,sutton2018,szepesvari2010}. Partially Observable Markov Decision Processes (POMDPs) extend these ideas to environments in which the true state of the system cannot be directly observed \citep{cassandra1998,kaelbling1998}. Bayesian decision theory, filtering theory, and hidden-state models provide complementary frameworks for representing uncertainty through posterior probability distributions and updating beliefs as information arrives \citep{berger1985,bernardo2000,jaynes2003,elliott1995,hamilton1989}. Information theory further provides quantitative measures of uncertainty, entropy, and information acquisition \citep{shannon1948,cover2006}. Collectively, these literatures suggest that autonomous AI systems should be viewed not as static predictive models but as partially observable decision processes.

This perspective is particularly familiar in quantitative finance. Portfolio managers, traders, and risk managers routinely make decisions without observing the true future state of the economy. Instead, they maintain probabilistic beliefs regarding latent economic conditions and act according to those beliefs. Modern portfolio theory, equilibrium asset pricing, active portfolio management, and Bayesian portfolio construction all operate within this decision-theoretic paradigm \citep{markowitz1952,sharpe1964,lintner1965,mossin1966,grinold2000,blacklitterman1992,dixon2020}. The resulting environment provides a natural setting in which to study the validation of autonomous agents.

This paper argues that agentic AI should be validated as a partially observable stochastic control system. We propose a model validation framework based on Partially Observable Markov Decision Processes in which posterior beliefs, forecasts, actions, and realized utility become explicit objects of validation. Large language models are interpreted as approximate Bayesian filtering operators acting on an information filtration, thereby transforming observations into probabilistic beliefs regarding latent states of the environment. The framework combines concepts from model risk management, calibration theory, stochastic control, Bayesian inference, and quantitative finance to create a practical methodology for validating autonomous decision systems.

The principal contribution is not a new agent architecture or portfolio strategy. Rather, the paper develops a layered validation framework that decomposes autonomous decision making into observations, beliefs, forecasts, actions, and utility. This decomposition permits independent validation of each stage of the decision process and allows failures of inference to be distinguished from failures of decision making. The framework is subsequently demonstrated through a portfolio-management case study using performance analysis, calibration diagnostics, belief-state coverage tests, ablation studies, and parameter-sensitivity analysis.

The remainder of the paper is organized as follows. Section 2 develops a POMDP representation of agentic AI and establishes the role of latent states, posterior beliefs, and utility functions. Section 3 introduces the proposed validation architecture. Section 4 discusses model risk. Section 5 presents a portfolio-management case study. Section 6 reports empirical validation results and finally, Section 7 concludes.

\section{POMDP Representation of Agentic AI}

The mathematical framework developed in this paper draws primarily upon the theory of stochastic control, Bayesian decision theory, partially observable Markov decision processes, filtering theory, and reinforcement learning \citep{bellman1957,bertsekas1995,puterman1994,kaelbling1998,cassandra1998,berger1985,bernardo2000,jaynes2003,elliott1995,sutton2018}. The central observation is that an autonomous AI system operates under uncertainty and therefore faces a decision problem rather than a purely predictive problem. At each point in time, the agent must act using incomplete information regarding the true state of the environment.

The standard POMDP framework provides a natural representation of this setting. The environment is described by a latent state process $S_t\in\mathcal S$, where $\mathcal S=\{s_1,\ldots,s_K\}$ denotes a finite state space. The agent receives observations $O_t\in\mathcal O$ and selects actions $A_t\in\mathcal A$. Information available to the agent evolves according to a filtration $\mathbb F=\{\mathcal F_t\}_{t\ge0}$ defined on a probability space $(\Omega,\mathcal F,\mathbb P)$.

The state process is assumed to satisfy the Markov property

\begin{equation}
\mathbb P(S_{t+1}\in B \mid S_t,S_{t-1},\ldots,S_0)
=
\mathbb P(S_{t+1}\in B \mid S_t)
\end{equation}
for all measurable sets $B\subseteq\mathcal S$. State transitions are governed by a transition kernel $T(s,s')=\mathbb P(S_{t+1}=s' \mid S_t=s)$, while observations are generated according to an observation kernel $Z(o,s)=\mathbb P(O_t=o\mid S_t=s)$.

The agent's objective is to maximize expected discounted utility. Given a policy $\pi$, the associated value function is

\begin{equation}
V^\pi
=
\mathbb E^\pi
\left[
\sum_{t=0}^{T}
\gamma^t
R(S_t,A_t)
\right],
\end{equation}
where $\gamma\in(0,1]$ denotes a discount factor. The quantity $R(S_t,A_t)$ denotes the one-period reward obtained when action $A_t$ is taken while the environment is in latent state $S_t$. The reward function encodes the objective of the decision maker and therefore provides the link between observations, decisions, and utility. Depending on the application, the reward may represent financial return, operational efficiency, user satisfaction, task completion, risk reduction, regulatory compliance, or any other measurable objective. The reward function is assumed to be measurable with respect to the state-action pair and may be either deterministic or stochastic. The optimal policy is therefore

\begin{equation}
\pi^\star
=
\arg\max_{\pi}
V^\pi.
\end{equation}
Unlike traditional machine-learning systems, the objective is therefore not prediction accuracy but expected utility.

\subsection{Belief States}

Since the latent state is not directly observable, the agent must maintain a probabilistic representation of uncertainty. Following the standard POMDP literature \citep{kaelbling1998,cassandra1998}, the posterior belief state is defined by

\begin{equation}
b_t(s)
=
\mathbb P
\left(
S_t=s
\mid
\mathcal F_t
\right).
\end{equation}
The belief vector $b_t=(b_t(s_1),\ldots,b_t(s_K))$ lies in the probability simplex

\begin{equation}
\Delta^{K-1}
=
\left\{
x\in\mathbb R^K:
x_i\ge0,\;
\sum_{i=1}^{K}x_i=1
\right\}.
\end{equation}
The belief state therefore provides a complete probabilistic summary of uncertainty regarding the hidden environment. Upon observing $O_{t+1}=o$, the posterior distribution evolves according to the Bayesian filtering recursion

\begin{equation}
b_{t+1}(s')
=
\frac{
Z(o,s')
\sum_{s\in\mathcal S}
T(s,s')
b_t(s)
}
{
\sum_{\tilde s\in\mathcal S}
Z(o,\tilde s)
\sum_{s\in\mathcal S}
T(s,\tilde s)
b_t(s)
}.
\end{equation}
This recursion forms the basis of the filtering problem and provides the connection between observations and posterior beliefs.

\subsection{Belief-State Sufficiency}

A well known fundamental result in POMDP theory is that the posterior belief state is a sufficient statistic for optimal decision making.

\begin{theorem}[Belief-State Sufficiency {\citep{puterman1994,kaelbling1998,cassandra1998}}]
Let $ H_t = (O_0,A_0,\ldots,O_t)$ denote the complete observation history. Then there exists an optimal policy $\pi^\star$ such that
$A_t
=
\pi^\star(b_t)$.
Equivalently, the optimal policy depends upon the history $H_t$ only through the posterior belief state $b_t$.
\end{theorem}

See Appendix A for proof. This theorem provides the principal motivation for belief-state validation. If optimal decisions depend only on posterior beliefs, then the quality of those beliefs becomes a primary object of validation.

\subsection{Value of Information}

Information plays a central role in autonomous decision systems. Let $\mathcal I_1$ and $\mathcal I_2$ denote two information structures satisfying $\mathcal I_1\subseteq\mathcal I_2$. Define the associated value functions by

\begin{equation}
V(\mathcal I_j)
=
\sup_{\pi\in\Pi(\mathcal I_j)}
\mathbb E^\pi
\left[
\sum_{t=0}^{T}
\gamma^tR_t
\right],
\qquad
j=1,2,
\end{equation}
where $\Pi(\mathcal I_j)$ denotes the set of admissible policies adapted to information structure $\mathcal I_j$.

\begin{theorem}[Non-Negativity of Information Value]
If $\mathcal I_1\subseteq\mathcal I_2$, then

\begin{equation}
V(\mathcal I_2)
\ge
V(\mathcal I_1).
\end{equation}
Consequently,

\begin{equation}
VOI(\mathcal I_2,\mathcal I_1)
=
V(\mathcal I_2)-V(\mathcal I_1)
\ge0.
\end{equation}

\end{theorem}

See Appendix A for proof. This theorem formalizes an important principle for information governance. Additional information cannot reduce the value attainable by an optimal decision maker. The relevant governance question is therefore not whether information is valuable, but whether its marginal value exceeds the associated costs, risks, and operational constraints.

\subsection{Implications for Validation}

The POMDP framework naturally decomposes autonomous decision making into a sequence of transformations,

\begin{equation}
\mathrm{Observations}
\rightarrow
\mathrm{Beliefs}
\rightarrow
\mathrm{Forecasts}
\rightarrow
\mathrm{Actions}
\rightarrow
\mathrm{Utility}.
\end{equation}

This decomposition is central to the validation framework developed in the remainder of the paper. Beliefs become the object of inference validation, forecasts become the object of predictive validation, actions become the object of policy validation, and utility becomes the object of performance validation. By separating these layers, failures of state estimation can be distinguished from failures of forecasting or decision making, thereby providing substantially greater transparency than validation based solely on final outputs.

The next section develops quantitative methodologies for validating each of these components.

\section{Validation Framework}

The central premise of this paper is that autonomous systems should be validated as sequential decision processes rather than predictive models. The POMDP framework introduced in Section 2 naturally decomposes agent behavior into a sequence of transformations,

\begin{equation}
\mathrm{Observations}
\rightarrow
\mathrm{Beliefs}
\rightarrow
\mathrm{Forecasts}
\rightarrow
\mathrm{Actions}
\rightarrow
\mathrm{Utility}.
\end{equation}

This decomposition provides the basis for a layered validation architecture in which beliefs, forecasts, policies, and utility are evaluated separately. Such separation is important because failures of state estimation, forecasting, decision making, and objective specification represent fundamentally different forms of model error.

\subsection{Belief-State Validation}

The posterior belief state is the primary object of inference validation. A useful belief process should be calibrated, meaning that events assigned probability \(p\) occur approximately \(p\) percent of the time in the long run. Calibration is evaluated using proper scoring rules. Let \(Y_t\) denote the realized state label and \(\widehat b_{t,k}\) the estimated probability of state \(k\). The multiclass Brier score is

\begin{equation}
BS
=
\frac{1}{T}
\sum_{t=1}^{T}
\sum_{k=1}^{K}
\left(
\widehat b_{t,k}
-
\mathbf 1_{\{Y_t=k\}}
\right)^2,
\end{equation}
while the logarithmic score is

\begin{equation}
LS
=
-
\frac{1}{T}
\sum_{t=1}^{T}
\log
\widehat b_{t,Y_t}.
\end{equation}
The logarithmic score is a strictly proper scoring rule and is uniquely minimized by the true conditional distribution \citep{dawid1982,gneiting2007}. Consequently, both measures provide objective assessments of belief quality independent of downstream actions. Belief states may also be evaluated using entropy and information gain,

\begin{equation}
H(b_t)
=
-
\sum_{k=1}^{K}
b_t(s_k)\log b_t(s_k),
\qquad
IG_t
=
H(b_{t-1})-H(b_t),
\end{equation}
which quantify the reduction in uncertainty produced by new information.

\subsection{Forecast Validation}

Forecast validation assesses whether posterior beliefs contain useful predictive information. Let \(\widehat\mu_t=f(\widehat b_t,X_t)\) denote a forecast generated from the estimated belief state. Forecast quality may be evaluated using prediction errors or, in financial applications, the information coefficient

\begin{equation}
IC
=
Corr(\widehat\mu_t,R_{t+1}),
\end{equation}
which measures the relationship between forecasts and realized outcomes.

\subsection{Policy Validation}

Policy validation evaluates the quality of decisions rather than predictions. If actions are selected according to \(A_t=\pi(\widehat b_t)\), policy quality is measured through the realized value

\begin{equation}
V^\pi
=
\mathbb E^\pi
\left[
\sum_{t=0}^{T}
\gamma^tR_t
\right].
\end{equation}
Relative to a benchmark policy \(\pi_b\), the incremental value generated by the agent is

\begin{equation}
\Delta V
=
V^\pi-V^{\pi_b}.
\end{equation}

\subsection{Utility Validation}

A policy may be internally consistent yet optimize an inappropriate objective. Utility validation therefore assesses whether the optimization objective aligns with organizational goals. This distinction becomes particularly important in autonomous systems because an agent may behave optimally with respect to a specified objective while remaining misaligned with the true decision-making objective.

\subsection{Validation Summary}

The validation framework therefore evaluates four distinct components: belief quality, forecast quality, policy quality, and utility quality. The empirical analysis implements this framework through calibration diagnostics, coverage tests, ablation studies, sensitivity analysis, and utility-based performance evaluation. Together, these diagnostics provide a structured methodology for validating autonomous decision systems beyond their final outputs.

\section{Model Risk}

The preceding section introduced a layered validation architecture. This section identifies the principal sources of model risk associated with each layer and discusses the assumptions underlying the proposed framework.

The central observation is that autonomous systems are hierarchical models. Information is transformed into beliefs, beliefs into forecasts, forecasts into actions, and actions into utility. Each transformation introduces uncertainty and therefore a distinct source of model risk.

\subsection{State-Space Risk}

The first source of model risk arises from the specification of the latent-state space. The framework assumes that the environment evolves on a finite set of hidden states

\begin{equation}
\mathcal S
=
\{
s_1,\ldots,s_K
\}.
\end{equation}

This assumption is central to the entire methodology because posterior beliefs, forecasts, and decisions are all conditioned upon the latent-state representation. If the state space fails to capture economically meaningful variation in the environment, then the resulting beliefs may be systematically misleading even when inference is otherwise accurate.

State-space risk therefore arises whenever the chosen latent states provide an incomplete or distorted representation of reality. Formally, if \(S_t^\star\) denotes the unknown true latent process and \(S_t\) denotes the model representation, then state-space error may be represented abstractly as $\varepsilon_S
=
S_t^\star-S_t$. In practice, the true latent state space is almost never observable and is generally unknown. Consequently, the state-space specification should be viewed as a modeling assumption rather than an empirical fact.

To illustrate this, suppose the true economic environment contains, say, two distinct inflationary regimes. The first corresponds to an orderly inflationary expansion characterized by strong growth and rising commodity prices. The second corresponds to stagflation, in which inflation is accompanied by deteriorating growth and weakening corporate earnings. If the latent-state model contains only a single state labelled \emph{Inflation Shock}, then these economically distinct environments are forced into the same category. The resulting posterior beliefs may therefore be internally consistent while still failing to capture important differences relevant to forecasting and decision making.

Similarly, a state labelled \emph{Crisis} may inadvertently absorb multiple forms of uncertainty, including financial crises, geopolitical shocks, recession fears, interest-rate volatility, and market stress. In such circumstances, elevated posterior probability assigned to the Crisis state may not necessarily indicate a genuine crisis. Instead, it may reflect deficiencies in the state-space representation itself. This phenomenon is observed empirically in Section 6, where the Crisis state receives persistently high probability despite relatively benign realized market outcomes.

State-space risk is closely related to model-specification risk in traditional model validation. In credit-risk models, for example, important risk drivers may be omitted from a probability-of-default model. In stress-testing frameworks, scenario definitions may fail to capture relevant economic mechanisms. In the present framework, the analogous concern is that the latent-state representation may omit important dimensions of the environment.

The latent states used in the empirical application are intentionally specified \emph{a priori} using economic intuition rather than statistical optimization. This choice improves interpretability and facilitates validation, but it also introduces the possibility of misspecification. Consequently, the latent-state representation should be viewed as a useful approximation rather than a literal description of the economy.

From a validation perspective, state-space risk is particularly important because it cannot be eliminated through improved calibration or more sophisticated forecasting models. If the underlying representation of the environment is incorrect, then all downstream components inherit that error. State-space validation should therefore be regarded as a prerequisite for belief-state validation, forecast validation, and policy validation.

\subsection{Filtering Risk}
The second source of model risk concerns state estimation. In a partially observable environment, the true latent state is not directly observable and must instead be inferred from available information. The ideal posterior belief state is

\begin{equation}
b_t
=
\mathbb P
\left(
S_t
\mid
\mathcal F_t
\right).
\end{equation}
In practice, however, the posterior distribution is generally unavailable because the information filtration may contain heterogeneous and unstructured data. Modern agentic systems frequently combine numerical observations, documents, emails, web content, retrieved knowledge, tool outputs, user interactions, market data, and intermediate reasoning traces. Constructing an explicit probabilistic model for such information is typically infeasible. The framework therefore employs an approximate filtering operator $\widehat{\Phi}_{\theta}$,
implemented using a large language model (LLM). The resulting posterior estimate is $\widehat b_t
=
\widehat{\Phi}_{\theta}
(
\mathcal F_t
)$.

The role of the LLM should be interpreted carefully. The model is not assumed to replace probability theory, Bayesian inference, or filtering theory. Rather, it serves as a flexible approximation to the conditional expectation operator that maps a complex information set into a probability distribution over latent states. Conceptually, the LLM acts as a semantic inference engine that compresses large volumes of structured and unstructured information into an interpretable posterior belief state.

This interpretation is consistent with recent developments in foundation models and autonomous agents. Modern frontier models are increasingly used as reasoning engines embedded within larger workflows that incorporate retrieval, planning, tool use, memory, and iterative decision making \citep{bommasani2021,park2023,yao2023react,shinn2023reflexion,xi2023,wang2024agents}. In such systems, the LLM does not operate in isolation. Instead, it is situated within a broader architecture in which information is gathered, processed, and transformed into actions.

A useful way to view contemporary agentic systems is through the filtration

\begin{equation}
\mathcal F_t
=
\sigma
\left(
X_t,
D_t,
T_t,
M_t
\right),
\end{equation}
where $X_t$ denotes structured observations, $D_t$ denotes retrieved documents and external knowledge, $T_t$ denotes tool outputs, and $M_t$ denotes memory or interaction history. Modern orchestration frameworks allow frontier models to access each of these information sources through retrieval-augmented generation, tool invocation, external APIs, database queries, search systems, and workflow engines \citep{park2023,yao2023react,wang2024agents}. The resulting information set is often substantially richer than the numerical state vectors traditionally encountered in filtering problems.

From a stochastic-control perspective, the LLM therefore plays a role analogous to an approximate Bayesian filter operating on an enlarged information filtration. The objective is not to predict the next observation directly, but to infer a posterior distribution over latent states that is useful for downstream decision making. The resulting filtering error is $\varepsilon_B
=
\widehat b_t
-
b_t$. This error captures the discrepancy between the inferred belief state and the ideal Bayesian posterior.

Several distinct sources contribute to filtering risk. First, the latent-state specification itself may be incomplete or misspecified. If the chosen state space fails to capture important dimensions of the environment, even a perfect inference engine cannot recover the correct posterior. Second, the information filtration may be incomplete. The agent can only condition upon information available through observations, retrieval systems, tools, and memory. Missing information necessarily limits the quality of posterior inference. Third, frontier models may exhibit reasoning errors, hallucinations, retrieval failures, or sensitivity to prompting strategies \citep{weidinger2022,bai2022constitutional,ganguli2022red,liang2022holistic}. Such failures may distort the inferred posterior distribution.Finall,y model updates and changing foundation-model behavior introduce a form of model drift. Since frontier models evolve over time, the mapping $\widehat{\Phi}_{\theta}$ may itself change, potentially altering posterior beliefs even when the underlying information set remains unchanged.

Filtering risk is particularly important because all subsequent forecasts and decisions depend upon the inferred belief state. An agent may possess an optimal forecasting model and an optimal decision policy, yet still perform poorly if the underlying beliefs are inaccurate. Consequently, belief-state validation becomes a first-class model-validation problem rather than merely a component of downstream predictive accuracy.

The calibration diagnostics, coverage tests, and ablation studies developed later in the empirical section of this paper are designed specifically to evaluate this source of model risk.

\subsection{Forecast Risk}

The third source of model risk concerns the mapping from beliefs to forecasts. Expected returns are represented by

\begin{equation}
\widehat\mu_t
=
f(\widehat b_t,X_t).
\end{equation}
The forecast error is $\varepsilon_F
=
\widehat\mu_t-\mu_t$, where \(\mu_t\) denotes the true conditional expectation. Forecast risk may remain substantial even when the inferred beliefs are well calibrated. Accurate state estimation does not guarantee accurate forecasting.

\subsection{Policy Risk}

The fourth source of model risk concerns decision making. The agent selects actions according to

\begin{equation}
A_t
=
\pi(\widehat b_t).
\end{equation}
The resulting policy may fail to maximize utility because of optimization constraints, estimation error, implementation choices, or deficiencies in the portfolio-construction methodology. If \(A_t^\star\) denotes the optimal action under the assumed utility function, policy error may be represented as

\begin{equation}
\varepsilon_P
=
A_t-A_t^\star.
\end{equation}

\subsection{Utility-Specification Risk}

A final source of model risk concerns the objective function itself. The framework assumes that decisions are evaluated according to a specified utility function \(\widehat U\). The true organizational objective, however, may differ from the objective embedded within the model. The utility-specification error is therefore $\varepsilon_U
=
U^\star-\widehat U$. Examples include regulatory objectives, reputational concerns, operational constraints, and risk tolerances that are not fully captured by the model.

\subsection{Parameter Risk}

The empirical implementation requires specification of parameters governing risk aversion, prior shrinkage, forecast weighting, and portfolio construction. Although these parameters are economically motivated, they remain uncertain. Model outputs may therefore depend upon parameter selection. To assess robustness, the empirical study performs sensitivity analysis across a range of plausible parameter values. The objective is not to identify a unique optimal specification but rather to determine whether conclusions remain stable under reasonable perturbations.

\subsection{Aggregate Model Risk}

The preceding discussion suggests that model risk in autonomous systems is inherently multi-dimensional. An aggregate representation may be written as

\begin{equation}
\varepsilon_{\mathrm{total}}
=
\varepsilon_S
+
\varepsilon_B
+
\varepsilon_F
+
\varepsilon_P
+
\varepsilon_U.
\end{equation}
This decomposition provides a practical framework for challenge, monitoring, stress testing, and governance. More importantly, it highlights that failures of autonomous systems need not originate from a single source. An apparently poor decision may arise from inaccurate beliefs, flawed forecasts, an inappropriate policy, or a misaligned objective function.

The portfolio-management case study presented in the following sections illustrates how these sources of model risk can be measured and evaluated in practice.

\section{Portfolio Management Case Study}

This section demonstrates the proposed validation framework in a portfolio-management setting. The objective is not to introduce a new portfolio-construction methodology, but rather to provide a realistic environment in which belief-state validation, forecast validation, policy validation, and utility validation can be studied simultaneously.

Portfolio management is particularly suitable for this purpose because decisions must be made under uncertainty using incomplete information regarding future economic conditions. The resulting problem is naturally consistent with the partially observable decision framework developed in the previous sections.

\subsection{Latent State Specification}

The empirical implementation employs the latent-state representation

\begin{equation}
\mathcal S
=
\{
\mathrm{AI\ Boom},
\mathrm{Soft\ Landing},
\mathrm{Inflation\ Shock},
\mathrm{Recession},
\mathrm{Crisis}
\}.
\end{equation}
These states were specified \emph{a priori} using economic intuition rather than statistical estimation. The objective is not to provide a complete representation of the economy but to construct an interpretable latent-state model that can be subjected to validation.

The states are intended to capture broad economic and market environments associated with growth, inflation, recession, financial stress, and technology-driven expansion. As discussed in Section 4, the resulting representation should be viewed as an approximation and is therefore subject to state-space risk.

\subsection{Information Set and State Inference}

The information filtration is constructed from both market and macroeconomic variables. Market variables include returns, volatility measures, momentum indicators, and cross-asset relationships. The macroeconomic information set includes measures of market volatility, commodity prices, and credit conditions. Posterior beliefs are estimated using the filtering framework introduced previously. In the baseline implementation, a large language model acts as an approximate filtering operator that maps the available information set into a posterior distribution over latent states. The resulting belief vector provides the input to the forecasting and decision-making layers.

The purpose of the macroeconomic variables is not to forecast returns directly. Rather, they improve observability of the latent state process and therefore contribute indirectly through the belief-state estimation procedure.

The practical implementation of the filtering step requires a mechanism for transforming the information filtration into a posterior distribution over latent states. The resulting query structure is described next.

\subsection{LLM Prompt Structure and State-Inference Query}

The posterior belief state is obtained by presenting the available information set to a frontier large language model and requesting a probability distribution over the latent-state space. The objective is not to predict future returns directly but rather to estimate the current hidden state of the environment.

At each decision date, the model receives a structured summary containing recent market behavior, macroeconomic information, and the previous belief state. The information supplied to the model includes recent asset returns, volatility measures, momentum indicators, market-implied volatility, commodity-price dynamics, credit-market conditions, and other variables contained in the filtration \(\mathcal F_t\).

The query is formulated as a state-estimation problem rather than a forecasting problem. Specifically, the model is instructed to infer the probability that the current environment belongs to each state in $\mathcal S$. The model is required to return a valid probability distribution

\begin{equation}
\widehat b_t
=
(\widehat b_{t,1},\ldots,\widehat b_{t,K}),
\end{equation}

satisfying

\begin{equation}
\widehat b_{t,k}\ge0,
\qquad
\sum_{k=1}^{K}
\widehat b_{t,k}
=
1.
\end{equation}
A simplified version of the prompt structure may be represented schematically as

\begin{equation}
\text{Market Features}
+
\text{Macro Features}
+
\text{Previous Beliefs}
\rightarrow
\text{LLM}
\rightarrow
\widehat b_t.
\end{equation}

The previous belief state is supplied to encourage temporal consistency and to provide the model with a probabilistic summary of the preceding environment. Consequently, the resulting inference process resembles a recursive filtering procedure rather than a sequence of independent classifications.

Importantly, the model is not asked to recommend trades or portfolio allocations at this stage. The purpose of the query is solely to estimate the latent state of the environment. This separation is critical from a validation perspective because it allows the quality of the inferred beliefs to be evaluated independently of subsequent forecasts and decisions.

A representative prompt sent to the Open AI API (aka ChatGPT) may be summarized as follows:

\begin{quote}
\it{Given the following market observations, macroeconomic indicators, and previous posterior beliefs, estimate the probability that the current environment corresponds to each latent state. Return probabilities for AI Boom, Soft Landing, Inflation Shock, Recession, and Crisis. Ensure that probabilities sum to one and provide a brief rationale for the allocation.}
\end{quote}

The resulting probability distribution becomes the input to the forecasting and portfolio-construction layers described below. Consequently, the empirical analysis can separately evaluate the quality of the inferred beliefs, the forecasts generated from those beliefs, and the decisions ultimately produced by the agent.

\subsection{Belief-Conditioned Forecasting}

Expected returns are generated conditional on the inferred belief state. Let

\begin{equation}
\widehat\mu_t
=
f(\widehat b_t,X_t)
\end{equation}
denote the forecast vector produced by the model. Unlike purely historical forecasting approaches, expected returns are allowed to vary as posterior beliefs evolve through time. Consequently, changes in the inferred latent state induce corresponding changes in forecasted asset returns. This construction allows the empirical study to evaluate whether latent-state information contributes incremental forecasting value beyond information already contained in observable market variables.

\subsection{Portfolio Construction}

Forecasts are incorporated into a Bayesian portfolio-construction framework based on the Black--Litterman methodology \citep{blacklitterman1992}. Equilibrium expected returns are combined with belief-conditioned views to obtain posterior expected returns $\mu_{BL}$. Portfolio weights are then obtained by solving a constrained mean--variance optimization problem,

\begin{equation}
w^\star
=
\arg\max_w
\left(
w^\top\mu_{BL}
-
\lambda
w^\top\Sigma w
\right),
\end{equation}
subject to portfolio constraints. The resulting allocation represents the action selected by the agent.

\subsection{Validation Objectives}

The portfolio application provides a concrete environment in which each component of the validation framework can be evaluated. Belief-state validation examines whether the inferred posterior probabilities are calibrated and informative. Forecast validation evaluates whether latent-state information improves expected-return estimates relative to simpler forecasting approaches. Policy validation assesses whether the resulting portfolio decisions improve risk-adjusted performance relative to benchmark strategies and finally, utility validation evaluates whether the complete decision process improves the agent's objective function. The empirical results presented in the next section combine calibration diagnostics, coverage tests, ablation studies, sensitivity analysis, and portfolio-performance evaluation to assess each of these layers independently.

\section{Empirical Validation Results}

This section evaluates the validation framework developed in Sections 2--5. The objective is not merely to compare portfolio strategies, but to determine whether latent-state inference produces measurable improvements in belief quality, forecasting performance, and decision quality. The empirical analysis therefore focuses on four questions:

\begin{enumerate}
\item Are the inferred belief states informative and reasonably calibrated?
\item Do latent-state beliefs improve forecasts relative to simpler models?
\item Do belief-conditioned decisions improve utility relative to benchmark strategies?
\item Are the resulting conclusions robust to alternative model specifications?
\end{enumerate}

Collectively, these questions correspond to the belief, forecast, policy, and utility layers of the validation architecture introduced earlier.

\subsection{Experimental Design}

Historical market data were obtained from the Massive.com API over the period
\(10~\mathrm{June}~2024\) to \(12~\mathrm{June}~2026\). The empirical study employs a deliberately simple asset universe consisting of

\begin{equation}
\{
\mathrm{AAPL},
\mathrm{MSFT},
\mathrm{GOOGL},
\mathrm{NVDA},
\mathrm{AMZN},
\mathrm{JPM},
\mathrm{IBM},
\mathrm{GLD},
\mathrm{TLT}
\},
\end{equation}
with SPY serving as the benchmark market portfolio. The objective is not to optimize a production investment strategy but to provide a controlled environment in which the proposed validation framework can be evaluated.

The complete Forecasting POMDP framework is compared with five benchmark portfolio-construction methodologies:

\begin{enumerate}
\item Equal Weight;
\item Risk Parity;
\item Maximum Sharpe;
\item 60/40 SPY--TLT;
\item POMDP Utility.
\end{enumerate}
The Forecasting POMDP corresponds to the full architecture described in Section 5, combining latent-state inference, belief-conditioned forecasting, Bayesian portfolio construction, and utility-based optimization.

\subsection{Parameter Specification}

The parameter values used throughout the study are reported in Table \ref{tab:parameters}.

\begin{table}[H]
\centering
\caption{Experimental parameter values.}
\label{tab:parameters}
\begin{tabular}{lr}
\toprule
Parameter & Value \\
\midrule
Lookback Window & 252 Trading Days \\
Holding Period & 21 Trading Days \\
Rebalancing Frequency & Monthly \\
Risk Aversion $\lambda$ & 3.50 \\
Maximum Asset Weight & 0.35 \\
Prior Shrinkage $\rho$ & 0.25 \\
View Weight $\alpha$ & 0.65 \\
Portfolio Constraint & Long Only \\
Budget Constraint & Fully Invested \\
\bottomrule
\end{tabular}
\end{table}
The lookback window corresponds approximately to one year of trading history, while the holding period corresponds to approximately one month. The risk-aversion coefficient
$\lambda=3.5$ was selected to produce volatility levels broadly consistent with institutional balanced portfolios. The prior-shrinkage parameter $\rho=0.25$ reflects the well-known instability of historical return estimates. The view-weight parameter $\alpha=0.65$ assigns substantial influence to latent-state forecasts while preserving a meaningful contribution from historical information. These parameter choices are varied later in this section to assess robustness of the validation methods under perturbations.

\subsection{Portfolio Performance Validation}

Table \ref{tab:performance} reports the primary performance statistics.

\begin{table}[htbp]
\caption{Performance comparison of POMDP portfolio agent and benchmarks.}
\label{tab:performance}
\centering
\resizebox{\textwidth}{!}{
\begin{tabular}{lrrrrrrr}
\toprule
Strategy & CAGR & Volatility & Sharpe & Sortino & Max Drawdown & Calmar & Terminal Wealth \\
\midrule
Forecasting POMDP & 0.1587 & 0.1092 & 1.6811 & 1.3514 & -0.0782 & 2.0296 & 1.2320 \\
POMDP Utility & 0.1818 & 0.1234 & 1.4977 & 1.7901 & -0.1011 & 1.7993 & 1.2670 \\
Equal Weight & 0.2307 & 0.1754 & 1.2587 & 2.4743 & -0.1180 & 1.9545 & 1.3419 \\
Risk Parity & 0.2307 & 0.1754 & 1.2587 & 2.4743 & -0.1180 & 1.9545 & 1.3419 \\
60 40 SPY/ TLT & 0.1145 & 0.1380 & 0.9494 & 1.7781 & -0.1065 & 1.0756 & 1.1660 \\
Max. Sharpe & 0.1013 & 0.1512 & 0.8698 & 1.0295 & -0.1389 & 0.7296 & 1.1465 \\
\bottomrule
\end{tabular}
}
\end{table}

Several observations emerge immediately. The Equal Weight portfolio achieves the highest compound annual growth rate. This result is largely attributable to the strong performance of technology equities during the sample period. However, the objective of the proposed framework is not to maximize raw return. The objective is to improve decision quality under uncertainty. The Forecasting POMDP achieves the highest Sharpe ratio, the highest Calmar ratio, and the smallest maximum drawdown among the strategies considered. These results suggest that latent-state inference improves the efficiency with which risk is transformed into return.

From a validation perspective, Table \ref{tab:performance} evaluates the policy layer of the framework. The table demonstrates that the actions generated by the agent produce superior risk-adjusted outcomes relative to benchmark decision rules.

\subsection{Utility Validation}

The most important validation criterion is utility. Table \ref{tab:utility} reports mean--variance utility values.

\begin{table}[htbp]
\caption{Mean-variance utility comparison of \newline the POMDP portfolio agent and benchmarks.}
\label{tab:utility}
\centering
\begin{tabular}{lr}
\toprule
Strategy & Mean Variance Utility \\
\midrule
Forecasting POMDP & 0.1418 \\
POMDP Utility & 0.1315 \\
Equal Weight & 0.1131 \\
Risk Parity & 0.1131 \\
60 40 SPY/TLT & 0.0644 \\
Max. Sharpe & 0.0515 \\
\bottomrule
\end{tabular}
\end{table}

The Forecasting POMDP achieves the highest utility among all strategies. This result is significant because utility corresponds directly to the optimization objective of a rational decision maker. The findings therefore support the central thesis of the paper: autonomous systems should be evaluated according to decision quality rather than prediction accuracy alone.

\subsection{Belief-State Validation}

Figure \ref{fig:belief} illustrates the evolution of posterior beliefs through time.

\begin{figure}[H]
\centering
\includegraphics[width=.9\textwidth]{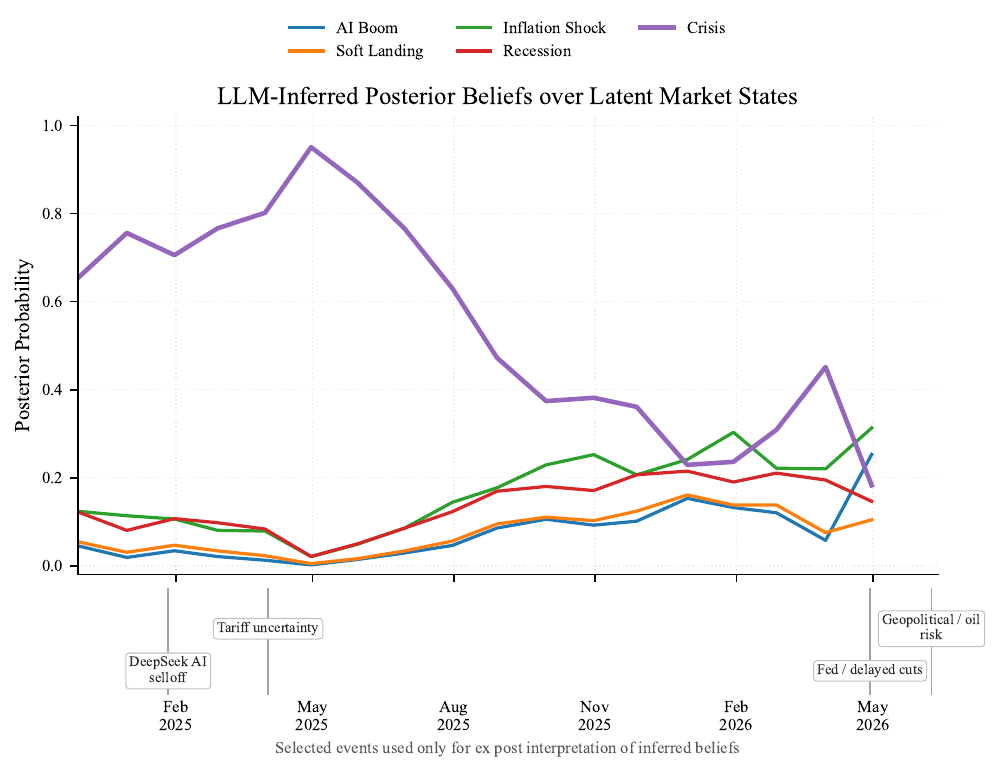}
\caption{Posterior belief states inferred by the filtering model.}
\label{fig:belief}
\end{figure}
The figure reveals substantial variation in posterior beliefs over the sample period. The Crisis state frequently receives elevated probability, indicating that the filtering model interprets volatility, credit conditions, and macroeconomic uncertainty as evidence of elevated systemic risk. From a validation perspective, the figure is important because it exposes the internal state of the agent. Traditional portfolio strategies reveal only final allocations. The present framework reveals the intermediate belief states responsible for those decisions. Consequently, failures of inference can be distinguished from failures of decision making.

\subsection{Belief Calibration Diagnostics}

A central contribution of the paper is the treatment of posterior beliefs as first-class validation objects. Table \ref{tab:belief_calibration} reports belief calibration diagnostics.

\begin{table}[H]
\caption{Belief calibration diagnostics using ex post proxy latent-state labels. The proxy labels are constructed from realized forward returns and are used only for validation diagnostics; they are not observed by the agent.}
\label{tab:belief_calibration}
\centering
\begin{tabular}{lrrrrr}
\toprule
State & Average Belief & Empirical Frequency & Calibration Gap & Brier Score & Log Score \\
\midrule
AI Boom & 0.074 & 0.278 & -0.204 & 0.255 & 0.884 \\
Soft Landing & 0.075 & 0.222 & -0.147 & 0.202 & 0.828 \\
Inflation Shock & 0.165 & 0.167 & -0.002 & 0.142 & 0.448 \\
Recession & 0.136 & 0.000 & 0.136 & 0.022 & 0.149 \\
Crisis & 0.550 & 0.333 & 0.216 & 0.336 & 0.944 \\
\bottomrule
\end{tabular}
\end{table}

Several observations emerge. The Inflation Shock state exhibits almost perfect calibration, with an average posterior probability of 0.165 and an empirical frequency of 0.167. This result suggests that the filtering framework is capable of identifying inflation-related environments with reasonable accuracy.

The Crisis state is systematically overestimated. The average posterior probability of 0.550 exceeds the empirical frequency of 0.333, producing a calibration gap of 0.216. This result is consistent with the belief trajectories shown in Figure \ref{fig:belief}, where the Crisis state appears to absorb multiple forms of uncertainty. Conversely, the AI Boom and Soft Landing states are systematically underrepresented. This finding suggests that the current latent-state specification may be overly conservative.

Because the latent states are not observable, calibration diagnostics therefore rely upon ex post proxy state labels derived from realized forward returns. The resulting statistics should be interpreted as approximate belief-validation measures rather than exact measures of posterior accuracy.

\subsection{Policy Validation Through Drawdowns}

Figure \ref{fig:drawdowns} reports portfolio drawdowns.

\begin{figure}[H]
\centering
\includegraphics[width=.9\textwidth]{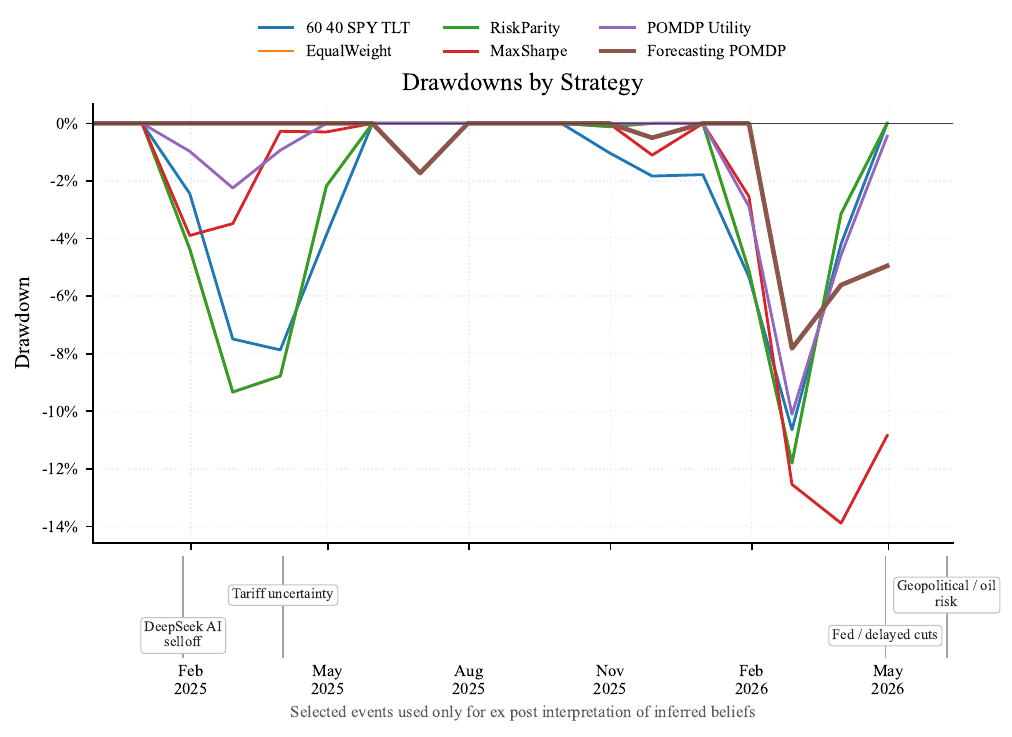}
\caption{Portfolio drawdowns.}
\label{fig:drawdowns}
\end{figure}
The Forecasting POMDP exhibits the smallest drawdown among all strategies considered. This result suggests that the latent-state framework contributes useful information regarding downside risk and adverse market environments. The figure therefore provides validation evidence for the policy layer of the framework.

\subsection{Policy Validation Through Wealth Evolution}

Figure \ref{fig:wealth} reports cumulative wealth trajectories.

\begin{figure}[H]
\centering
\includegraphics[width=.9\textwidth]{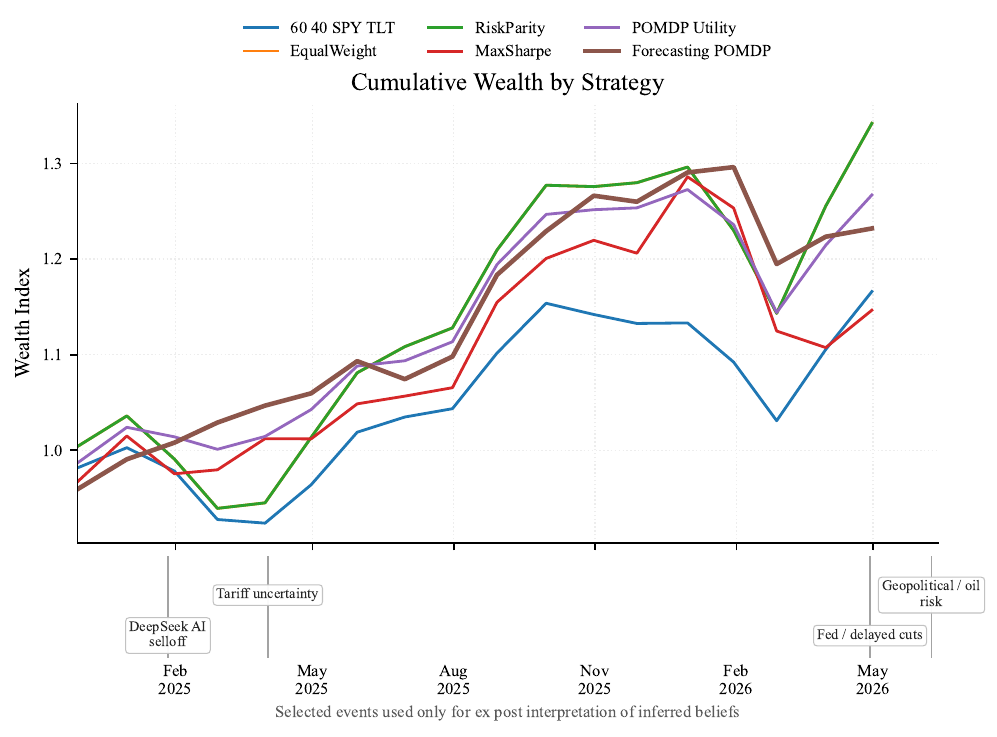}
\caption{Cumulative wealth trajectories.}
\label{fig:wealth}
\end{figure}
Although the equal weight portfolio achieves the highest terminal wealth, the Forecasting POMDP remains competitive while producing substantially superior risk-adjusted performance. This distinction illustrates the difference between maximizing return and maximizing utility.

\subsection{Ablation Study}

Ablation analysis provides one of the most important validation tests in the paper. Table \ref{tab:ablation} reports the performance of several reduced versions of the framework.

\begin{table}[H]
\caption{Ablation study showing the contribution of macro variables, LLM-inferred beliefs, and the full Forecasting POMDP framework.}
\label{tab:ablation}
\centering
\begin{tabular}{lrr}
\toprule
Model & Sharpe & Utility \\
\midrule
Historical Only & 1.427 & 0.113 \\
Market Only & 1.598 & 0.133 \\
Market + Direct Macro & 1.459 & 0.122 \\
Market + Beliefs & 1.681 & 0.142 \\
Full POMDP: Macro-Inferred Beliefs & 1.681 & 0.142 \\
\bottomrule
\end{tabular}
\end{table}

The results support several conclusions. First, market-based forecasting signals improve performance relative to historical-return estimates alone. Second, direct incorporation of macroeconomic variables into the forecasting model reduces performance, suggesting that macro variables are difficult to translate directly into expected-return forecasts. Third, latent-state beliefs improve both Sharpe ratio and utility relative to market-only forecasts. Most importantly, the Full POMDP and Market + Beliefs models outperform all alternatives. This result suggests that macroeconomic information is most useful when employed to improve latent-state inference rather than when introduced directly into the forecasting model. From a validation perspective, the ablation study provides evidence that the belief-state layer contributes independently to decision quality and is not merely duplicating information already contained in observable market variables.

\subsection{Sensitivity Analysis}

Table \ref{tab:sensitivity} reports parameter-sensitivity results.

\begin{table}[H]
\caption{Parameter sensitivity analysis for the Forecasting POMDP. Each row varies one parameter while holding all other parameters fixed at their base values.}
\label{tab:sensitivity}
\centering
\begin{tabular}{lllrrr}
\toprule
Parameter & Value & Base Value & Sharpe & Utility & Max Drawdown \\
\midrule
BASE MODEL &  &  & 1.681 & 0.142 & -0.078 \\
risk aversion & 2.000 & 3.500 & 1.749 & 0.152 & -0.079 \\
risk aversion & 3.500 & 3.500 & 1.681 & 0.142 & -0.078 \\
risk aversion & 5.000 & 3.500 & 1.643 & 0.137 & -0.078 \\
prior shrinkage & 0.100 & 0.250 & 1.692 & 0.143 & -0.078 \\
prior shrinkage & 0.250 & 0.250 & 1.681 & 0.142 & -0.078 \\
prior shrinkage & 0.500 & 0.250 & 1.663 & 0.140 & -0.079 \\
view weight & 0.400 & 0.650 & 1.582 & 0.128 & -0.077 \\
view weight & 0.650 & 0.650 & 1.681 & 0.142 & -0.078 \\
view weight & 0.800 & 0.650 & 1.743 & 0.151 & -0.079 \\
\bottomrule
\end{tabular}
\end{table}

The results indicate that the principal conclusions remain stable across a broad range of parameter values. Increasing the view weight improves both Sharpe ratio and utility, providing additional evidence that latent-state forecasts contribute useful information. Also, increasing prior shrinkage slightly reduces performance, suggesting that latent-state information contains signal not fully reflected in historical returns. We also observed that increasing risk aversion produces a gradual reduction in utility and Sharpe ratio, as expected. Importantly, performance remains relatively stable across all parameter configurations examined. This result reduces concerns regarding parameter overfitting and supports the robustness of the framework.

\begin{table}[H]
\caption{Kupiec-style unconditional coverage test for inferred belief states using ex post proxy latent-state labels. The null hypothesis is that the average inferred belief probability equals the empirical proxy-state frequency.}
\label{tab:belief_coverage}
\resizebox{\textwidth}{!}{
\begin{tabular}{lrrrrrrrl}
\toprule
State & Observations & Observed Count & Expected Count & Observed Rate & Expected Rate & Coverage Statistic & p-value & Reject 5\% \\
\midrule
AI Boom & 18 & 5 & 1.327 & 0.278 & 0.074 & 10.972 & 0.001 & Yes \\
Soft Landing & 18 & 4 & 1.351 & 0.222 & 0.075 & 5.616 & 0.018 & Yes \\
Inflation Shock & 18 & 3 & 2.969 & 0.167 & 0.165 & 0.000 & 0.984 & No \\
Recession & 18 & 0 & 2.457 & 0.000 & 0.136 & 2.845 & 0.092 & No \\
Crisis & 18 & 6 & 9.896 & 0.333 & 0.550 & 3.408 & 0.065 & No \\
\bottomrule
\end{tabular}
}
\end{table}

Table \ref{tab:belief_coverage} reports Kupiec-style coverage diagnostics for the inferred belief states using ex post proxy latent-state labels. The Inflation Shock state exhibits excellent coverage properties, while the AI Boom and Soft Landing states are systematically underrepresented. Conversely, the Crisis state is overrepresented, consistent with the belief-state trajectories observed in Figure \ref{fig:belief}. These findings suggest that the filtering model tends to interpret uncertainty conservatively, assigning excessive probability to adverse states while underestimating growth-oriented regimes. Because the latent states are not directly observable and the sample contains only eighteen independent decision periods, the results should be interpreted as exploratory validation diagnostics rather than definitive calibration tests.

\subsection{Validation Assumptions}
The empirical validation procedures rely upon several assumptions. First, the latent states are not directly observable. Consequently, belief calibration cannot be evaluated against the true hidden state. Instead, calibration and coverage diagnostics employ ex post proxy state labels constructed from realized market outcomes. These diagnostics should therefore be interpreted as approximate measures of belief quality rather than exact tests of posterior correctness.

Second, the portfolio-management application is intended as a proof-of-concept validation environment rather than a production investment strategy. The objective is to evaluate the validation framework itself rather than to maximize investment performance. Third, the ablation study assumes that improvements in utility can be attributed to the incremental addition of information and model components. While this approach is commonly used in machine learning and model validation, the resulting comparisons should be interpreted as evidence of contribution rather than definitive causal attribution. Fourth, the sensitivity analysis assumes that robustness to parameter perturbations provides evidence against excessive dependence on a particular calibration. The objective is not to identify a unique optimal parameter specification but rather to assess whether the principal conclusions remain stable under reasonable changes in model assumptions.

Finally, all empirical results are obtained from a single application domain and a relatively limited sample period. Consequently, the results should be interpreted as validation evidence supporting the proposed framework rather than proof of universal applicability.

\subsection{Validation Interpretation}

Taken together, the empirical results support the central thesis of the paper. Figure \ref{fig:belief} validates the belief layer by exposing the posterior state estimates used by the agent. Table \ref{tab:belief_calibration} provides direct evidence regarding the quality of those beliefs. Table \ref{tab:ablation} validates the forecasting layer by demonstrating that latent-state information contributes independently to decision quality. Figures \ref{fig:drawdowns} and \ref{fig:wealth} validate the policy layer by evaluating the consequences of the resulting decisions. Finally, Tables \ref{tab:performance} and \ref{tab:utility} validate the utility layer by measuring realized decision quality.

The principal contribution of the empirical study is therefore not the portfolio strategy itself. Rather, it demonstrates how an autonomous decision system can be decomposed into beliefs, forecasts, actions, and utility, with each component subjected to independent validation. This decomposition provides substantially greater transparency than evaluating agentic AI systems solely through their final outputs.

\section{Conclusion}

This paper has proposed a model validation framework for agentic AI based on the theory of Partially Observable Markov Decision Processes. The central premise is that autonomous systems should not be evaluated solely through the quality of their outputs. Instead, validation should focus on the complete decision-making process, including information acquisition, belief formation, forecasting, action selection, and realized utility.

The POMDP framework provides a natural mathematical representation of this process. By introducing latent states, posterior beliefs, and utility-based decision making, the framework extends traditional model validation concepts beyond prediction accuracy and toward the validation of autonomous decision systems operating under uncertainty.

A key contribution of the paper is the introduction of a layered validation architecture. Belief states are evaluated through calibration and scoring rules. Forecasts are evaluated through predictive performance. Policies are evaluated through realized decisions and risk-adjusted outcomes. Utility functions are evaluated through their ability to achieve organizational objectives. This decomposition enables model validators to identify whether failures originate from state-space specification, filtering, forecasting, policy construction, or utility design.

The paper also develops an explicit model-risk taxonomy for agentic AI. State-space risk, filtering risk, forecast risk, policy risk, utility-specification risk, and parameter risk are identified as distinct sources of model uncertainty. This taxonomy provides a practical framework for challenge, governance, monitoring, and ongoing validation of autonomous systems.

The empirical portfolio-management application serves as a proof-of-concept implementation of the framework. The results demonstrate that latent-state beliefs improve decision quality relative to traditional benchmark strategies. The ablation study provides evidence that belief-state inference contributes independently to performance, while the sensitivity analysis indicates that the principal conclusions are robust to reasonable perturbations of model parameters. The calibration diagnostics further demonstrate how posterior beliefs themselves may become objects of validation.

Perhaps the most important finding is methodological rather than empirical. The experiments illustrate that agentic AI systems can be decomposed into a sequence of interpretable transformations,

\begin{equation}
\mathrm{Observations}
\rightarrow
\mathrm{Beliefs}
\rightarrow
\mathrm{Forecasts}
\rightarrow
\mathrm{Actions}
\rightarrow
\mathrm{Utility},
\end{equation}

and that each stage can be validated independently. This capability provides a degree of transparency that is generally unavailable in traditional black-box evaluations.

The framework proposed here is intentionally general. Although the empirical study focuses on portfolio management, the methodology is equally applicable to enterprise copilots, autonomous research systems, cyber-defense agents, procurement systems, and other forms of agentic AI. In each case, the fundamental challenge remains the same: decisions must be made under uncertainty using incomplete information.

Several directions for future research remain. First, richer latent-state models could be learned directly from data rather than specified \emph{a priori}. Second, more sophisticated calibration procedures could be developed for settings in which latent states are not directly observable. Third, the framework could be extended to multi-agent systems, hierarchical agents, and environments with strategic interaction. Finally, the relationship between belief-state validation and emerging AI governance requirements warrants further investigation.

More broadly, the results suggest that quantitative finance, Bayesian decision theory, and stochastic control provide a mature mathematical foundation for the validation of autonomous AI systems. As organizations increasingly deploy agentic AI in high-consequence environments, the central question is unlikely to be whether an agent occasionally produces the correct answer. The more important question is whether the agent forms calibrated beliefs, uses information appropriately, makes defensible decisions, and improves expected utility. The framework developed in this paper provides a formal methodology for addressing that question.

\appendix

\section{Proofs}

\subsection{Belief-State Sufficiency}
\begin{proof}
For any measurable function $f$ of future states and observations,

\begin{equation}
\mathbb E
\left[
f(S_{t+1},S_{t+2},\ldots)
\mid
H_t
\right]
=
\int
f
\,dP(\cdot\mid H_t).
\end{equation}
By Bayes' rule and the Markov property,

\begin{equation}
P(S_{t+1}=s'\mid H_t)
=
\sum_{s\in\mathcal S}
T(s,s')
b_t(s).
\end{equation}
Thus the conditional distribution of all future states depends on $H_t$ only through $b_t$. Consequently, two histories that induce the same posterior belief state generate identical distributions for all future rewards. The Bellman recursion may therefore be written on the belief space:

\begin{equation}
V(b)
=
\max_{a\in\mathcal A}
\left\{
R(b,a)
+
\gamma
\sum_{o\in\mathcal O}
P(o\mid b,a)
V(\tau(b,a,o))
\right\},
\end{equation}
where $\tau$ denotes the Bayesian belief-update operator. Since the value function depends only on $b$, an optimal policy may also be expressed solely as a function of $b$. Hence the posterior belief state is a sufficient statistic.

\end{proof}

\subsection{Non-negativity of Information Value}
\begin{proof}

Since $\mathcal I_1\subseteq\mathcal I_2$, every policy adapted to $\mathcal I_1$ is also adapted to $\mathcal I_2$. Therefore

\begin{equation}
\Pi(\mathcal I_1)
\subseteq
\Pi(\mathcal I_2).
\end{equation}
Taking the supremum of the same objective function over a larger feasible set cannot decrease the optimum. Hence

\begin{equation}
\sup_{\pi\in\Pi(\mathcal I_2)}
J(\pi)
\ge
\sup_{\pi\in\Pi(\mathcal I_1)}
J(\pi), ~\textrm{where}~J(\pi)
=
\mathbb E^\pi
\left[
\sum_{t=0}^{T}
\gamma^tR_t
\right].
\end{equation}
The result follows immediately.
\end{proof}

\end{document}